%% file: Main.tex
\documentclass[a4paper,UKenglish,cleveref, autoref, thm-restate,numberwithinsect]{lipics-v2021}

\pdfoutput=1 %
\hideLIPIcs  %

\usepackage{amsmath}
\usepackage{amsthm}
\usepackage{amssymb}
\usepackage{hyperref}

\usepackage{xspace}
\usepackage{caption}
\usepackage{xcolor}
\usepackage{tikz}
\usetikzlibrary{positioning}
\usepackage{booktabs}
\usepackage[capitalise]{cleveref}

\usepackage{algorithm}
\usepackage{algpseudocode}
\algrenewcommand\algorithmicrequire{\textbf{Input:}}
\algrenewcommand\algorithmicensure{\textbf{At every step $t$:}}
\algrenewcommand\algorithmicdo{\textbf{Initialize:}}
\newcommand{\cost}{\mathrm{cost}}

\newcommand{\cP}{\mathcal{P}}
\newcommand{\cT}{\mathcal{T}}

\newcommand{\brho}{\boldsymbol{\rho}}

\newcommand{\OPT}{\mathrm{OPT}}

\newcommand{\cm}{\left\lceil \frac{m}{2}\right\rceil}

\allowdisplaybreaks

\usepackage[appendix=append, bibliography=common]{apxproof}
\theoremstyle{plain}
\newtheoremrep{claimapp}[theorem]{Claim}
\theoremstyle{plain}
\newtheoremrep{theoremapp}[theorem]{Theorem}
\theoremstyle{plain}
\newtheoremrep{lemmaapp}[theorem]{Lemma}
\newtheorem{question}{Question}
\newtheorem{assumption}[theorem]{Assumption}
\title{Improved Algorithms and Lower Bounds for Parametrized Metrical Service Systems} 

\author{Junhao Gan}{School of Computing and Information Systems, The University of Melbourne, Australia}{junhao.gan@unimelb.edu.au}{https://orcid.org/0000-0001-9101-1503}{Supported in part by the Australian Government through the Australian Research Council ARC DP230102908.}

\author{Xiao Sun}{School of Computing and Information Systems, The University of Melbourne, Australia}{xiao.sun.1@student.unimelb.edu.au}{https://orcid.org/0009-0003-2635-8764}{Supported by the Australian Government through the Australian Research Council DP240101353 and by the University of Melbourne through the Melbourne Research Scholarship.}

\author{Seeun William Umboh}{School of Computing and Information Systems, The University of Melbourne, Australia \and ARC Training Centre in Optimisation Technologies, Integrated Methodologies, and Applications (OPTIMA), Australia}{william.umboh@unimelb.edu.au}{https://orcid.org/0000-0001-6984-4007}{Supported by the Australian Government through the Australian Research Council DP240101353.}

\authorrunning{J. Gan, X. Sun, S. W. Umboh}

\Copyright{Junhao Gan, Xiao Sun, Seeun William Umboh} %

\ccsdesc[500]{Theory of computation~Online algorithms}

\keywords{online algorithms, competitive analysis, metrical service systems, metrical task systems}

\acknowledgements{We thank Christian Coester for insight into the case of $m=3$ on general metrics.}

\nolinenumbers %

\begin{document}

\maketitle

\begin{abstract}
We consider the parametrized setting of the classical metrical service system (MSS) problem first studied by Bubeck and Rabani (APPROX/RANDOM 2020). In this setting, the adversary is restricted to a set of $m$ distinct request types, known to the algorithm in advance. The goal is to obtain competitive ratio bounds in terms of $m$. In this work, we make significant progress in understanding the landscape of parametrized MSS and resolve several open problems from Bubeck and Rabani.

Our first main result is a tight bound for parametrized MSS on weighted stars. Previously, Bubeck and Rabani gave a randomized lower bound of $\Omega(m)$ and deterministic upper bound of $O(2^m)$. We show that, surprisingly,   a deterministic $O(m)$-competitive algorithm exists, matching the randomized lower bound. Our key insight is an interval covering formulation of MSS on weighted stars which enables an application of the primal-dual method.

Our second main contribution is an improved lower bound construction for parametrized MSS on hierarchically separated trees (HSTs). Bubeck and Rabani's construction gave a $\omega(1)$ lower bound when $m \geq 6$. Our improved lower bounds are tight for $2$-level HSTs and also rule out $O(1)$-competitive algorithms on HSTs when the parameter $m\geq 4$. We also complement these results by giving a deterministic $O(1)$-competitive algorithm on general metrics when $m=2$ while showing that it is impossible when $m\geq 3$.  
\end{abstract}

\input{intro}
\input{preliminaries}
\input{weighted-star}

\input{hst}
\begin{toappendix}
\input{two-sets}
\end{toappendix}
\section{Open Problems}

Following the initial study of parametrized metrical task systems (MTS) and metrical service systems (MSS) by Bubeck and Rabani~\cite{DBLP:conf/approx/BubeckR20}, our research answers several open problems posed by Bubeck and Rabani. We provide a deterministic $O(m)$-competitive algorithm on weighted stars, tighten the gap between the upper and lower bounds of randomized competitive ratios on HSTs of different levels, and show that it is possible to design a strictly constant-competitive algorithm on general metrics when $m=2$. However, multiple problems remain open: as \cref{ques3} has been resolved on general metrics by showing the contrast between $m=2$ and $m=3$ cases, the focus is now on special classes of metrics, e.g., is it possible to obtain a finite competitive ratio for chasing $3$ sets on HSTs or line metrics? Another interesting open problem is whether the $O(\binom{m}{\lceil m/2\rceil})$ upper bound of randomized algorithms on $2$-level HSTs can also be achieved by a deterministic algorithm. The progress on this problem will not only close the gap between the upper and lower bounds known for deterministic algorithms for this particular class of metrics, but also deepen our fundamental understanding of the parametrized MSS problem, as it concerns whether deterministic algorithms are essentially as powerful as randomized algorithms on $2$-level HSTs and other classes of metrics. Moreover, the technique that we use in \cref{Section3} is a  linear programming primal-dual method that exploits the properties of parametrized MSS with $m$ request types, and we wonder whether this framework can be generalized to prove the results of the more general parametrized MTS on uniform metrics (the tight deterministic bound of $\Theta(m\log(en/m))$ was proven by Bubeck and Rabani) or weighted stars (no result yet).

\bibliography{Reference}

\end{document}

%% file: intro.tex
\section{Introduction}

The \emph{metrical task system} (MTS) problem introduced by Borodin et al.~\cite{DBLP:conf/stoc/BorodinLS87} is a general framework that models a wide range of online computing problems, e.g., paging~\cite{DBLP:journals/cacm/SleatorT85} and $k$-server~\cite{DBLP:journals/jal/ManasseMS90}. In this problem, one server moves between different points (also called states) in an $n$-point metric space to process a sequence of requests that arrive one by one; each request incurs different costs when the server is at different states. The total cost is the sum of the transition cost of moving between different states and the cost of serving the requests. An important special case is the \emph{metrical service system} (MSS) problem~\cite{DBLP:conf/dimacs/ChrobakL91} where every request has cost either $0$ or infinity everywhere in the whole metric space. Note that paging and $k$-server are both MSS problems.

Since its introduction almost four decades ago, MTS has been intensively studied on both general metrics and some special classes of metrics. On general metrics, in a series of recent breakthrough results, Bubeck et al.~\cite{DBLP:journals/siamcomp/BubeckCLL21} showed that the randomized competitive ratio is $O(\log^2 n)$ and Bubeck et al.~\cite{DBLP:conf/stoc/BubeckCR23} gave a matching \emph{existential lower bound}\footnote{The existential lower bound means that there exist $n$-point metrics on which the lower bound holds.}, 
refuting the previous widely-believed conjecture of the existence of a logarithmic upper bound. Bubeck et al.~\cite{DBLP:conf/stoc/BubeckCR23} also showed a \emph{universal lower bound}\footnote{The universal lower bound means that it holds for every $n$-point metric space.} of $\Omega(\log n)$, which is also tight in the sense that an $O(\log n)$ upper bound can be achieved on some special metrics, e.g. uniform metrics~\cite{DBLP:conf/stoc/BorodinLS87} and weighted stars~\cite{DBLP:journals/siamcomp/BubeckCLL21}. While MSS is a special case of MTS, the best competitive ratio that can be achieved by both deterministic and randomized algorithms for MSS in all cases is asymptotically the same as that of MTS only up to a constant factor.

While many online problems are metrical task systems, the best bounds on their competitive ratios are much lower than those implied by the general worst-case MTS bounds. For instance, while the $k$-server problem in an $N$-point uniform metric space can be reduced to an MSS instance with $n = \binom{N}{k}$ states, the randomized competitive ratio is known to be $O(\log k)$, which is much lower than the lower bound of $\Omega(\log\binom{N}{k})$ implied by the universal lower bound for MSS with $\binom{N}{k}$ states. The reason for this discrepancy is because while there are $2^n - 1$ possible requests in a general $n$-state MSS, there are only $N$ possible requests for a $k$-server. 

The above example motivates studying the competitive ratio of MTS in terms of a parameter of the set of possible requests, instead of the number of states $n$. Burley and Irani~\cite{DBLP:journals/algorithmica/BurleyI97} proposed examining request sequences drawn from a given set. For uniform metrics, they provided a deterministic algorithm with competitive ratio $O(\log n)$ times the (unknown) best possible competitive ratio. Another example that shows the power of restricting the set of possible requests is the power management problem~\cite{DBLP:journals/tecs/IraniSG03} in which there are only two possible requests and the competitive ratio is $O(1)$~\cite{DBLP:journals/siamcomp/AugustineIS08}. For MSS, one parameter that has been studied is the \emph{width}  (e.g.~\cite{DBLP:conf/focs/FiatFKRRV91,DBLP:conf/soda/Ramesh93,DBLP:journals/jal/Burley96}), which is the maximum size of every requested set. Width is well-motivated as  MSS can also be viewed as online set chasing, in the sense that the server needs to move to a state in the feasible set where the cost of the current request is $0$ rather than $\infty$, and we only need to consider transition cost while the service cost is zero if the server stays in the feasible set at every step.

\subparagraph{Parametrized MTS.} The parametrized MTS proposed by Bubeck and Rabani~\cite{DBLP:conf/approx/BubeckR20} studies the number of possible requests as the parameter. In particular, the adversary can only generate requests from a fixed set of $m$ requests known to the algorithm in advance. On uniform metrics, they showed that the deterministic competitive ratio is $\Theta(m\log (en/m))$, which is significantly better than the general bound of $\Theta(n)$. Interestingly, the randomized competitive ratio on uniform metrics is still $\Theta(\log n)$ even when $m = 2$. On the other hand, for two-level hierarchically separated tree (HST) metrics, the deterministic competitive ratio is $\Theta(n)$ and the randomized competitive ratio is $\Theta(\log n)$, even when $m$ is constant. In summary, restricting $m$ only improves the deterministic competitive ratio for uniform metrics, and does not help for two-level HSTs.

\subparagraph{Parametrized MSS.} The more interesting case  is parametrized MSS, where parametrization leads to substantial improvements of the competitive ratio from the general case. For MSS, restricting possible requests within a fixed set of size $m$ improves both deterministic and randomized competitive ratios: for uniform metrics, the deterministic and randomized competitive ratios are both $\Theta(m)$ and thus significantly better when $n\gg m$. Another interesting aspect of this result is that deterministic algorithms are essentially as powerful as randomized algorithms, which refutes the folklore that the randomized competitive ratio is polylogarithmic in the deterministic competitive ratio for most problems that can be interpreted as metrical task systems. As for the case of two-level HSTs, there are still gaps between the currently known upper and lower bounds: the deterministic competitive ratio is between $\Omega(\binom{m}{\lfloor m/2\rfloor})=\Omega(2^{m}/\sqrt{m})$ and $O(m2^{m})$ and the randomized competitive ratio is at least $\Omega(2^{\lfloor m/2\rfloor})$.

Bubeck and Rabani~\cite{DBLP:conf/approx/BubeckR20} proposed the following open problems in the study of parametrized MSS, which involve not only the exact tight bound of the competitive ratio in certain scenarios, but also understanding how restricting the number of possible requests $m$ can help improve the competitive ratio, e.g., can we devise an online algorithm whose competitive ratio is purely in terms of $m$ rather than metric space size $n$, and can we do this for some special classes of metrics like HSTs or even general metrics?

\begin{question}\label{ques1}
What is the deterministic competitive ratio on weighted star metrics?
\end{question}

\begin{remark}It is easy to design an $O(2^{m})$-competitive deterministic algorithm and an $O(m)$-competitive randomized algorithm by reducing the problem to MSS on a weighted star where there are only $2^{m}-1$ different states and applying the state-of-the-art result for MTS on weighted stars. For the special case of uniform metrics, the deterministic and randomized competitive ratios are both known to be $\Theta(m)$ so it is natural to ask if the deterministic algorithm can be generalized to  weighted stars.
\end{remark}

\begin{question}\label{ques2}
What is the tight bound of the worst-case competitive ratio on all HSTs of two and higher levels?
\end{question}

\begin{question}\label{ques3}
For what constant values of $m$ can we get constant-competitive algorithms for general metrics or some special classes of metrics like HSTs?
\end{question}

\subsection{Our Contributions}

In this paper, we make substantial progress towards understanding the competitiveness of parametrized MSS.

\subparagraph{Weighted stars.} Our first main result answers \cref{ques1} affirmatively.

\begin{theorem}
\label{star1}
    There is a $2m$-competitive deterministic algorithm for parameterized MSS on weighted stars.
\end{theorem}

The result brings more insight into the foundation of this problem and the power of randomization as well, because it matches the lower bound of $\Omega(m)$ for randomized algorithms since the uniform metric is a special class of the weighted star. \cref{star1} shows that deterministic algorithms are essentially as powerful as randomized algorithms in achieving the best possible competitive ratio of $O(m)$ on weighted stars, the same as the case for the parametrized metrical service system problem on uniform metrics.

\subparagraph{$2$-level HSTs.} Next, we address \cref{ques2} by settling the randomized competitive ratio on $2$-level HSTs. Let $c_{k,m}$ denote the best possible randomized competitive ratio of the MSS problem that can be achieved on all $k$-level HSTs with $m$ request types:

\begin{theorem}
\label{twohst}
    $c_{2,m}=\Theta( \binom{m}{\lceil m/2\rceil})=\Theta({2^{m}}/\sqrt m)$.
\end{theorem}

This theorem contains two separate claims: $c_{2,m}=O( \binom{m}{\lceil m/2 \rceil})$ (\cref{up2}) and $c_{2,m}=\Omega( \binom{m}{\lceil m/2\rceil}) $ (Theorem~\ref{lower2}). The lower bound is existential, in the sense that there exists a $2$-level HST and $m$ request types on it such that no randomized algorithm can achieve a better worst-case competitive ratio on all possible request sequences. \cref{twohst} improves upon previous upper and lower bounds. Previously, the lower bound of $\Omega( \binom{m}{\lceil m/2\rceil})$ is only known to hold for deterministic algorithms, and we prove that it also holds for randomized algorithms using a new hard instance.

\subparagraph{Higher level HSTs.} Finally, we turn to  \cref{ques3}. Previously, Bubeck and Rabani~\cite{DBLP:conf/approx/BubeckR20} obtained lower bounds on $c_{k,m}$ for $k\geq 3$ by iterating their lower bound construction for $k = 2$. As a consequence, their lower bound rules out $O(1)$-competitive algorithms for $m \geq 6$. By iterating our construction in the proof of \cref{lower2} to higher level HSTs, we also obtain improved lower bounds on $c_{k,m}$ for $k \geq 3$. The formal theorem statement (\cref{lowerboundG}) is somewhat technical so we defer it to \cref{sec:generalhst} and instead state the following corollary.

\begin{corollary}\label{imposs}
    For $m \geq 4$, there is no algorithm that achieves constant competitive ratio on HSTs.
\end{corollary}

Combined with \cref{m2}, we conclude that the only open case for Question~\ref{ques3} on HSTs is $m=3$.
\subparagraph{General metrics.}

We resolve \cref{ques3} on general metrics by the following two complement results.

We first show that a constant competitive ratio cannot be achieved on general metrics when $m$ is at least $3$ by the following existential lower bound of $\Omega(\log n)$:

\begin{theoremapprep}\label{m3}
    For $m \geq 3$, there exist $3n$-point metrics on which every randomized algorithm has competitive ratio of $\Omega(\log n)$.        
\end{theoremapprep}

\begin{proof}
Our proof is an adaptation of the proof that the randomized competitive ratio for parametrized MTS with $m=2$ on an $n$-point metric space is $\Omega(\log n)$ (Lemma 10 in~\cite{DBLP:conf/approx/BubeckR20}).

\subparagraph{Metrical structure.} Suppose we have $n$ points $\{v_{0},\ldots,v_{n-1}\}$ where $n=2^q$ for some positive integer $q$ and $d(v_{i}, v_{j})=1$ for every $i\neq j$. For each $j\in\{0,1,\ldots,n-1\}$, define two new points $v^0_j$ and $v^1_j$ such that $d(v_j,v^0_j)=C^{j-n}/2$ and $d(v_j,v^{1}_j)=C^{-j-1}/2$, where $C$ is a value to be determined later. The distance $d$ on the whole space $\cup_{j=1}^{n-1}\{v_j,v^0_j, v^1_j\mid 0\leq j\leq n-1\}$ is defined as the metric closure of the partial distance that we have defined. There are three types of requests for parametrized MSS, corresponding to three subsets of the metric: $K_{0}=\{v^0_j\mid 0\leq j\leq n-1\}, K_{1}=\{v^1_j\mid 0\leq j\leq n-1\}, K_{2}=\{v_j\mid 0\leq j\leq n-1\}$. For every $0\leq j\leq n-1$, we say that $\{v_j,v^0_j,v^1_j\}$ forms a cluster, then the distance between every two points from different clusters is at least $1$. Assume that the initial state of the MSS server is in $K_{2}$.

\subparagraph{Randomized request sequence.} The request sequence is generated in rounds, and it is sufficient to describe the instance for a single round as this process can always be repeated. Each round consists of $q$ phases numbered with $0,\ldots,q-1$ respectively. For every $0\leq i\leq q-1$, $h_{i}\in\{0,1\}$ is generated uniformly at random at the beginning of phase $i$. Given $h_0,\ldots,h_i$ that have been generated so far, let $a_i=(h_0\ldots h_i1\ldots1)_2$ be the maximum integer from $\{0,\ldots,n-1\}$ whose binary representation has $h_0\ldots h_i$ as a prefix if $h_i=0$, otherwise let $a_i=(h_0\ldots h_i0\ldots0)_2$ be the minimum integer from $\{0,\ldots,n-1\}$ whose binary representation has $h_0\ldots h_i$ as a prefix if $h_i=1$. During phase $i$, if $h_{i}=0$, then the sequence $K_0K_2$ is repeated for $C^{n-a_i-1}$ times; if $h_{i}=1$, then the sequence $K_{1}K_{2}$ is repeated for $C^{a_i}$ times.

\subparagraph{Optimal cost.} Let $h$ be the integer whose binary representation is $(h_{0}\ldots h_{q-1})_2$. The offline solution can move to $v_{h}$ at the beginning of the round and travel back and forth between $v_{h}$ and $v^{h_i}_h$ during each phase $i$. It pays cost at most $1$ of moving to $v_{h}$ from another point in $V$. The cost of moving back and forth between $v_{a_{i}}$ and $v^{h_i}_{a_{i}}$ is exactly $1/C$ for both $h_{i}=0$ and $h_{i}=1$. We observe that $h\leq a_i$ if $h_i=0$ and $h\geq a_i$ if $h_i=1$. As a result, we can verify that the cost of moving back and forth between $v_{h}$ and $v^{h_i}_h$ during phase $i$ is at most $1/C$ for every $0\leq i\leq q-1$, since $d(v_{j}, v^{0}_{j})$ is increasing in $j$ and $d(v_{j}, v^{1}_{j})$ is decreasing in $j$. In summary, the offline optimal cost is no more than $1+q/C$ during the whole round.

\subparagraph{Expected cost of every deterministic algorithm.} Now we consider the expected cost of every deterministic algorithm. For every $i$, define intervals of integers $J^{0}_{i}=[0,(h_0\ldots h_{i-1}01\ldots1)_2]$ and $J^{1}_{i}=[(h_0\ldots h_{i-1}10\ldots0)_2,n-1]$, then $J^{0}_{i}\cup J^1_i=\{0,\ldots,n-1\}$ and $J^0_i\cap J^1_i =\varnothing$ for every $i$. Suppose that the algorithm is at state $v_j$ at the beginning of phase $i$. Since $\{J^{0}_i, J^{1}_i\}$ is a partition of $\{0,\ldots, n-1\}$, the event that $j\in J^{1-h_i}_i$ happens with probability $1/2$ when $h_{i}\in\{0,1\}$ is generated uniformly at random. We prove that the cost during phase $i$ is at least $1$ as long as $j\in J^{1-h_i}_i$. If it moves to a point from another cluster during phase $i$, then it incurs cost of at least $1$. Otherwise, it moves back and forth between $v_j$ and $v^i_j$ during the entire phase $i$. Otherwise, if $h_i=0$ and it moves back and forth between $v_j^0$ and $v_j$, the cost incurred is $C^{j-a_i-1}\cdot$. If $j\in J^{1}_i$, then $j\geq a_i+1$ and the cost is at least $1$. Similarly, if $h_i=1$ and $j\in J^{0}_i$, then $j\leq a_i-1$ and the cost of moving back and forth between $v_j$ and $v^{1}_j$ during phase $i$ is at least $1$ as well. Therefore, the expected cost of every deterministic algorithm is at least $1/2$ during each phase $i$ and $q/2$ during the whole round. According to Yao's minimax principle, the worst-case competitive ratio of every randomized algorithm is at least $\frac{q}{2(1+q/C)}$, which is equal to $q/4$ if we set $C=q$. Based on the setting that $n=2^q$, the competitive ratio is $\Omega(\log n)$ and hence cannot be a constant for every such metric.
\end{proof}

\noindent The proof of \cref{m3} can be found in the Appendix.

On the other hand, a constant competitive ratio can be achieved on general metrics when $m=2$.

\begin{observation}\label{m2}
    When $m = 2$, there is a $6$-competitive deterministic algorithm on general metrics.
\end{observation}

This result follows from the observation that when $m=2$, we can assume that the input sequence is strictly alternating between two request types, because MSS can be regarded as chasing feasible sets and hence repeating the same request does not incur additional cost. Therefore, the only uncertainty of the input is its total length, and the optimal cost can be expressed as a non-decreasing function of the input length. This observation allows us to apply the Guess-and-Double Method to design a constant-competitive algorithm for general metrics when $m=2$. The proof details of \cref{m2} can be found in the Appendix.

Therefore, the open problems for \cref{ques3} now lie on some special classes of metrics, including HSTs that we study in this paper.

\subsection{Our Techniques}

\subparagraph{Weighted stars.}
In MSS, we can think of a request as the set of states that the server must move to. 
Bubeck and Rabani~\cite{DBLP:conf/approx/BubeckR20} use an algorithm that works in phases similar to the classical marking algorithm to provide a $O(m)$-deterministic algorithm for parametrized MSS on uniform metrics: in each phase, it moves to an arbitrary state in the intersection of the requests in the phase; if the intersection is empty, the algorithm starts a new phase and moves to an arbitrary requested state. It is clear that any solution needs to move at least once during each phase, and they show that the algorithm moves at most $m$ times during each phase to prove the upper bound. This method is heavily based on the fact that the metric space is uniform and hence we can simply compare the number of times the algorithm's server moves with that of the optimal server. %
It is unclear whether their method can be adapted to solve the problem on weighted stars.

We use a significantly different approach based on an interval covering formulation of the problem and apply the primal-dual method.  %
Each interval corresponds to a consecutive subsequence of requests, and the cost to buy this interval is the minimum possible cost of moving to a state that can satisfy all request types in this subsequence. %
Our analysis exploits the fact that there are only $m$ different request types to show that the algorithmic cost is at most $m$ times the feasible dual solution that we maintain and thus obtain the desired competitive ratio. Our formulation is similar to that in \cite{DBLP:journals/mor/BuchbinderN09} but we exploit the property of MSS requests that allows us to only consider the unique state with minimum cost to satisfy every subsequence, while for MTS we have to consider various costs of serving the same subsequence of requests in $n$ different states, thus leading to the appearance of $n$ in the competitive ratio instead of purely in terms of $m$. We hope that similar methods under the linear programming and primal-dual framework can be adapted to solve the more general parametrized MTS problem.

\subparagraph{Lower bounds on HSTs.}
We use the same high-level approach used by Bubeck and Rabani~\cite{DBLP:conf/approx/BubeckR20} in which a construction for the uniform metric\footnote{A $1$-level HST is equivalent to a uniform metric.} is lifted to HSTs of higher levels. The main conceptual difference is in our construction for the uniform metric. In Bubeck and Rabani's construction, there are $2^{m/2}$ points, each labeled by a node of the Boolean hypercube $\{0,1\}^{m/2}$ (we assume $m$ is even here for simplicity). Each of the $m$ requests correspond to a an axis-aligned halfspace of the Boolean hypercube. Thus, if we construct a sequence of $m/2$ requests by picking a random halfspace for each dimension with equal probability, then for each request, the algorithm has to move with probability $1/2$. On the other hand, the adversary can move to the point in the intersection of the halfspaces. 
This construction can be lifted to HSTs of higher levels by partitioning the set of all possible labels of the lower level case into pairs of antipodes and treating them as complementary requests.

Our key idea is to use a construction for the uniform metric that has more points. In particular, we use $\binom{m}{m/2}$ points, each corresponding to an $\frac{m}{2}$-sized subset of $[m]$. As a result, when we apply the above lifting process, for every fixed level $k$ of the HST, there are more subtrees in our construction when $m$ is sufficiently large. We need to be more careful in our analysis as it is no longer true that each request causes the algorithm to move with constant probability. Instead, we prove that its average is still lower bounded by a positive constant.

%% file: preliminaries.tex
\section{Preliminaries}

Given a metric space $\mathcal{M}=(X,d)$ where $X$ is the set of points and $d: X\times X\rightarrow \mathbb{R}_{\geq 0}$ is the distance function, a parametrized metrical service system instance consists of an initial location $s_0 \in X$ of the server and a sequence of requests $\{\rho_{t}\}_{t\geq 1}$, where each $\rho_{t}$ is from a fixed collection $\{K_{1},\ldots, K_{m}\}$ of $m$ non-empty subsets of $X$. Upon the arrival of a request $\rho_{t}=K_{i}$ for some $1\leq i\leq m$ at every step $t\geq 1$, one must move the server to some state $s_{t}\in K_{i}$ to maintain a feasible solution if its previous state $s_{t-1}$ does not satisfy this, and the decision must be made without knowing future requests after $\rho_{t}$. As a result, MSS can be regarded as the online set chasing problem, illustrating that we always need to move the server to a state in the current requested set, and the parametrized setting restricts that there are only $m$ different request types. The cost of a solution $\{s_{t}\}_{t\geq 1}$is defined to be
\begin{equation*}
    \sum_{t\geq 1} d(s_{t-1},s_{t}),
\end{equation*}
where $d(s_{t-1},s_{t})$ denotes the transition cost to move from $s_{t-1}$ to $s_{t}$ upon the arrival of $\rho_{t}$. This can be interpreted as the total length of the trajectory of the server.

A randomized algorithm $\mathcal{A}$ has competitive ratio $c$ if there exists some constant $\beta$ that is independent of the input sequence (it might depend on the metric space $\mathcal{M}$) such that for every input sequence $\boldsymbol{\rho}$:
\begin{equation*}
    \mathbb{E}[\cost_{\mathcal{A}}(\boldsymbol{\rho})]\leq c\cdot \cost^{*}(\boldsymbol{\rho})+\beta,
\end{equation*}
where $\mathbb{E}[\cost_{\mathcal{A}}(\boldsymbol{\rho})]$ is the expected cost of $\mathcal{A}$ on input $\boldsymbol{\rho}$ and $\cost^{*}(\boldsymbol{\rho})$ is the minimum cost among all feasible solutions of $\boldsymbol{\rho}$. If $\beta=0$, then we say that $c$ is a strict competitive ratio of $\mathcal{A}$. When we talk about the competitive ratio of a problem class (e.g. parametrized MSS with $m$ requests on a class of metrics), we are referring to the best possible competitive ratio that can be achieved by some algorithm $\mathcal{A}$ on all possible inputs from this problem class.

%% file: weighted-star.tex
\section{Weighted Stars}\label{Section3}

On every weighted star metric, there is one root node to which every leaf node (state) is connected with a single edge of non-negative weight. Let $d_{s}$ denote the weight of the edge between state $s$ and the root, then the distance between two different states $s_{1}$ and $s_{2}$ is $d(s_{1},s_{2})=d_{s_{1}}+d_{s_{2}}$, where $d_{s_{1}}$ can be charged as the cost of leaving state $s_{1}$ and $d_{s_{2}}$ can be charged as that of entering state $s_{2}$ if the server moves from $s_{1}$ to $s_{2}$.

A common practice on weighted stars is that instead of charging the cost of both leaving $s_{1}$ and entering $s_{2}$, we only need to consider the cost of entering the new state at every step, i.e. we can regard the cost of moving from $s_{1}$ to a different state $s_{2}$ as only $d_{s_{2}}$, for both the algorithmic cost and the optimal solution. Over the entire process, the number of times that the server leaves any state is the same as the number of times that the server enters the same state, except for the initial and the last state. Now we formally show that this setting will only lose the competitive ratio of the original problem setting by a factor of at most $2$.

\begin{lemma}\label{factor2}
    Every strictly $c$-competitive algorithm for the problem setting under the new transition costs has strict competitive ratio at most $2c$ in the original setting.
\end{lemma}

\begin{proof}
 Suppose that the initial state of the server is $s_{0}$, the algorithm moves to states $s_{1},\ldots,s_{n}$ by the chronological order during its execution, and some fixed solution moves to states $s_{1}',\ldots,s'_{m}$ one by one. By the fact that the algorithm achieves strict competitive ratio $c$ under the new setting, it holds that $\sum_{j=1}^{n}d_{s_{j}}\leq c\sum_{j=1}^{m}d_{s'_{j}}$. Then we have $d_{s_{0}}+2\sum_{j=1}^{n-1}d_{s_{j}}+d_{s_{n}}\leq d_{s_{0}}+2\sum_{j=1}^{n}d_{s_{j}}\leq d_{s_{0}}+2c\sum_{j=1}^{m}d_{s'_{j}}\leq 2c(d_{s_{0}}+2\sum_{j=1}^{m-1}d_{s'_{j}}+d_{s'_{m}})$, which suggests that the same algorithm achieves strict competitive ratio $2c$ in the original setting by the arbitrariness of the solution that we are comparing the algorithmic cost with.   
\end{proof}

\subparagraph{State labels.} In this setting, we do not need to consider the initial state. Moreover, each state can be labeled with a subset of $[m]=\{1,2,\ldots,m\}$, which implies the set of request types that are feasible at this position. To formalize, given the label $S\subseteq[m]$ of some state $s$ and $\{K_{1},\ldots,K_{m}\}$ of the parametrized MSS setting of $m$ possible request types as subsets of the metric space, $S=\{i\mid s\in K_{i}\}$. We may use the notations of a state $s$ and its label $S$ interchangeably, as well as a set of natural numbers and the set of request types $K_{i}$'s whose indices are in this set. 

\subparagraph{Assumptions.} It suffices to consider $2^{m}-1$ states that correspond to $2^{m}-1$ different non-empty subsets of $[m]$, since if there are at least two states with exactly the same label, then any reasonable algorithm should only visit the one whose edge to the node has the minimum weight. For every non-empty $S \subseteq [m]=\{1,2,\ldots,m\}$, we use $d_{S}$ to denote the minimum edge weight between the root and any node whose label is exactly equal to $S$. If no such node exists, just set $d_{S}=\infty$.

We now make another assumption, which is without loss of generality, yet important for the validity of the algorithm that we are going to propose.

\begin{assumption}\label{inputassumption}
    $d_{S}\leq d_{S'}$ for every $S\subseteq S'$.
\end{assumption}
Otherwise, if there exist two different non-empty subsets $S$ and $S'$ of $[m]$ such that $S\subseteq S'$ yet $d_{S}>d_{S'}$, then every visit to state $S$ can be replaced by a visit to state $S'$, maintaining the feasibility because $S'\supseteq S$ while only reducing the cost. In this case, we can set $d_{S}$ to $d_{S'}$ instead, which does not increase the optimal cost since we are only reducing edge lengths, while the algorithmic solution can replace each visit to $S$ with one visit to $S'$ in the original setting, which means that every algorithm can be transformed into one whose cost in the original setting is no more than that under this assumption.

We make one final assumption. Against any deterministic algorithm, one can assume that the adversary always forces the algorithm to move to a different state at every step as long as there is some $K_{i}$ such that the current state of the algorithm is not in $K_{i}$, and the input ends when the current state is in the intersection of all $K_{i}$'s for $1\leq i\leq m$, if such a scenario exists. A straightforward consequence of this assumption on our problem setting is that any two consecutive requests from the input sequence are of different types, but we will eventually use a stronger argument that also exploits the algorithmic decisions seen by the adversary in the deterministic model to bound the desired competitive ratio.

\subparagraph{Integer programming formulation.} Under the input sequence $\{\rho_{t}\}_{t\geq 1}$, we use $S_{i,j}$ to denote the set of request types that appear in the subsequence $\rho_{i},\ldots, \rho_{j}$, i.e. $S_{i,j}=\{\ell\mid K_{\ell}=\rho_{t} \text{ for some }i\leq t\leq j\}$, which also corresponds to a unique node of the weighted star in our notation. Then we can write down an integer programming formulation of the problem whose minimum objective function value is a lower bound of the optimal solution of the original problem. We use $T$ to express the number of steps of the whole input sequence.

\begin{equation}\label{linear}
\begin{array}{ll@{}ll}
\text{minimize}  & \displaystyle\sum\limits_{1\leq i\leq j\leq T}d_{S_{i,j}}x_{i,j} &\\
\text{subject to}& \displaystyle\sum\limits_{i\leq t\leq j}x_{i,j}\geq 1,  &\forall 1\leq t \leq T\\
                 &          x_{i,j} \in \{0,1\}, &\forall 1\leq i\leq j\leq T
\end{array}
\end{equation}

\begin{lemma}
    The minimum value of (\ref{linear}) is no more than the optimal value of the parametrized MSS problem on the weighted star in our setting.
\end{lemma}

\begin{proof}
    Given any solution of the parametrized MSS problem on the weighted star, for every state that the server stays at from step $i$ to $j$, we can just assume that it stays at state $S_{i,j}$ due to Assumption~\ref{inputassumption}, otherwise the cost can only be higher. We set $x_{i,j}=1$ for every such state that the server visits, then this is a feasible solution of (\ref{linear}), and its value is the same as the given solution of the parametrized MSS on the weighted star instance.
\end{proof}

The next step is to write down the dual program of the linear relaxation of (\ref{linear}), whose maximum value is no more than the minimum value of its primal linear program and thus the minimum value of the original program (\ref{linear}) as well:

\begin{equation}\label{dual}
\begin{array}{ll@{}ll}
\text{maximize}  & \displaystyle\sum\limits_{1\leq t\leq T}y_{t} &\\
\text{subject to}& \displaystyle\sum\limits_{t:i\leq t\leq j}y_{t}\leq d_{S_{i,j}},  &\forall 1\leq i\leq j\leq T\\
                 &          y_{t}\geq 0, &\forall 1\leq t\leq T
\end{array}
\end{equation}

\subparagraph{Primal-dual algorithm.} We are now ready to state our primal-dual algorithm. Our algorithm maintains a feasible dual solution and at every step $t$: increases the dual variable $y_{t}$ until some dual constraint becomes tight, and then decides which state to move to depending on the tight dual constraint. 

Upon the arrival of request $\rho_{t}$, increase the value of the dual variable $y_{t}$ from the initial value $0$ until some dual constraint is tight, i.e.
\begin{equation*}
    \sum_{\ell=i}^{t}y_{\ell}=d_{S_{i,t}}
\end{equation*}
for some $i\leq \ell\leq t$; or equivalently, we can write the value as
\begin{equation}\label{primaldualequation}
    y_{t}=\min_{1\leq i \leq t} \left\{d_{S_{i,t}}-\sum_{\ell=i}^{t-1}y_{\ell}\right\},
\end{equation}
where the value of $\sum_{\ell=i}^{t-1}y_{\ell}$ is defined to be zero for $i=t$.

If the tight dual constraint is $\sum_{\ell=i}^{t}y_{\ell}=d_{S_{i,t}}$, then the algorithm moves the server to the state labeled with $S_{i,t}$ (the one with the minimum edge weight to the root based on our previous assumption), to process request $\rho_{t}$. Ties are broken in an arbitrary way.

\subparagraph{Algorithm analysis.} We first prove that the algorithm indeed maintains a feasible dual solution.

\begin{lemma}
  The algorithm maintains a feasible dual solution.
\end{lemma}
\begin{proof}
    At every step $t$, a new dual variable $y_{t}$ whose initial value is zero appears with a set of new dual constraints, and we need to show that these constraints are not violated when $y_{t}=0$, so that it is feasible to increase the value of $y_{t}$ from zero until some dual constraint becomes tight.
    
    It suffices to prove that $\sum_{\ell=i}^{t-1}y_{\ell}\leq d_{S_{i,t}}$ for every $i\leq t$, given the values $y_{1}$ to $y_{t-1}$ well-defined from previous steps. The case that $i=t$ is trivial, since $d_{S_{i,t}}$ is non-negative based on the problem definition. For every $1\leq i\leq t-1$, it holds that $\sum_{\ell=i}^{t-1}y_{\ell}\leq d_{S_{i,t-1}}$ given that we maintain a feasible solution for dual variables and constraints that have appeared until step $t-1$, and the result follows by Assumption~\ref{inputassumption} and the fact that $S_{i,t-1}\subseteq S_{i,t}$.
\end{proof}

Now we are ready to use values of $y_{t}$'s that the algorithm computes under its execution to provide an upper bound of its cost.

\begin{lemma}
\label{lem:alg-ub}
    The cost of the algorithm is at most $m\sum_{\ell=1}^{T}y_{\ell}$.
\end{lemma}

\begin{proof}
    When the algorithm moves to state $S_{i,t}$ at step $t$, $S_{i,t}$ is the set of request types in some subsequence $\rho_{i},\ldots,\rho_{t}$ and the values of dual variables assigned by the algorithm so far satisfy that $\sum_{\ell=i}^{t}y_{\ell}=d_{S_{i,t}}$, since the algorithm chooses state $S_{i,t}$ to move to according to the tight dual constraint at every step. As a result, the cost of moving to state $S_{i,t}$ at every step $t$ can be charged to some $y_{\ell}$'s. Now we show that every $y_{\ell}$ is only charged to at no more than $m$ different steps $t$ and their corresponding states $S_{i,t}$'s, which means that we need to show there are at most $m$ pairs of $(i,t)$'s such that $i\leq \ell\leq t$ and the algorithm moves to state $S_{i,t}$ at step $t$.

    Suppose there are two different steps $t_{1}<t_{2}$ such that the algorithm moves to state $S_{i_{1},t_{1}}$ at step $t_{1}$ and state $S_{i_{2},t_{2}}$ at step $t_{2}$, and $i_{1}\leq \ell\leq t_{1}$ and $i_{2}\leq \ell\leq t_{2}$. The subsequence $\rho_{t_{1}+1},\ldots,\rho_{t_{2}}$ contains at least one request that is not in set $S_{i_{1},t_{1}}$---namely, $\rho_{t_1+1}$---as otherwise the algorithm has no incentive to move to a different state. As a result, since $S_{t_{1}+1,t_{2}}\subseteq S_{\ell,t_{2}}$ and $S_{\ell,t_{1}}\subseteq S_{i_{1},t_{1}}$, we  get $S_{\ell,t_{2}}\setminus S_{\ell,t_{1}}\supseteq S_{t_{1}+1,t_{2}}\setminus S_{i_{1},t_{1}}\neq\varnothing$,  and hence $|S_{\ell,t_{2}}|\geq |S_{\ell,t_{1}}|+1$ because $S_{\ell,t_{1}}\subseteq S_{\ell,t_{2}}$. Now if we have in total $k$ such different steps $\ell\leq t_{1}<\ldots<t_{k}$, then it holds that $|S_{\ell,t_{k}}|\geq k$, which directly implies that $k\leq m$.
\end{proof}

\cref{factor2,lem:alg-ub} and the fact that the maximum value of the linear dual program (\ref{dual}) is no more than the minimum value of the original program (\ref{linear}) imply Theorem~\ref{star1}.

%% file: hst.tex
\section{Randomized Competitive Ratio on HSTs}

In this section, we consider hierarchically separated trees (HSTs), a highly symmetric rooted tree structure where the edge weights between a node and all its children are equal. The path between one internal node to all leaves in the subtree rooted at this node has the same number of edges, meaning that the level of a node as the number of edges from it to any of its descendant leaf is well-defined (the level of every leaf node is $0$). Moreover, the edge weight between an internal node and all its children is a function of the node level. Since we regard all leaf nodes as the points in the metric space induced by a rooted tree, the distance between two points is a function of the level of their lowest common ancestor. We say that an HST is a $k$-level HST if its root is at level $k$.

The bulk of this section is devoted to proving our lower bounds in \cref{iterativesection}. We finish with the tight upper bound for 2-level HSTs in \cref{2-tight}. 

\subsection{Iterative Lower Bounds}\label{iterativesection}
The key ingredient for our improved lower bounds is a new construction for uniform metrics (i.e.~$1$-level HSTs). Thus, we begin in \cref{lb-unif} by giving the construction and analysis for uniform metrics. Then, we show how to lift our ideas to construct  higher level HSTs in \cref{lb-lift}. In \cref{lb-2}, we prove our improved lower bound $c_{2,m}=\Omega(\binom{m}{ \lceil m/2\rceil})$ on $2$-level HSTs. Finally, we prove our lower bounds for higher level HSTs.

\subparagraph{Iterative sequence.} We first introduce the following iterative sequence which will be used to express the constructions of hard instances and the lower bounds:
\begin{equation}
    \label{eq:iter}
A_{1,m}=\binom{m}{\lceil m/2\rceil},
    \qquad
    A_{k,m}=\binom{A_{k-1,m}}{\lceil A_{k-1,m}/2\rceil}
    \quad \mbox{for  $k\geq 2$}.
\end{equation}

\begin{proposition}
    \label{iterdiv}
    For every fixed $m \geq 4$, we have $A_{1,m} > m$ and so the sequence $\{A_{k,m}\}_{k \geq 1}$ is unbounded. 
\end{proposition}

\subparagraph{Labels and requests.} 
In this section, the requests are defined by assigning a label $S \subseteq [m]$ to each leaf node. The request type $K_i$ (as a feasible set) is then defined to be the leaf nodes whose labels contain $i$, i.e.~request type $K_i$ is served by moving to a node $v$ whose label $S$ contains $i$. We may use number $i$ to refer to request type $K_i$ when the context is clear, especially when we use the corresponding subset of $[m]$ to refer to a set of request types.

\subparagraph{Randomized lower bounds.} To prove a lower bound of $\alpha$ on the competitive ratio of randomized online algorithms, we
use Yao’s minimax principle~\cite{DBLP:conf/focs/Yao77}, and give a distribution over input instances such that every deterministic algorithm has expected competitive ratio at least $\alpha$. 

In the following, we will only describe and analyze a randomized input sequence whose length is a function of $k$ and $m$. This process can be repeated to generate arbitrarily many rounds and hence the request sequence can be arbitrarily long and have arbitrarily large cost.

\subsubsection{Warm-up: Uniform Metrics}
\label{lb-unif}
We begin with a warm-up on a $1$-level HST (uniform metric) which we will lift to higher level HSTs.

\subparagraph{Lower bound instance.} The key property of our $1$-level HST is that for every request sequence that is a permutation of $[m]$, there is a leaf that can serve the first $\cm$ requests and another leaf that can serve the remaining requests, and thus there is always a solution that moves at most twice.
Define the $1$-level HST structure $\cT_{1,m}$ as follows: it has $\binom{m}{\lceil m/2\rceil }=A_{1,m}$ leaves, each labeled with a $\cm$-size subset of $[m]$ and the distance between every two leaves is $1$. The adversarial request sequence is a uniform random permutation of $[m]$. 

Before we start our analysis, we prove the following proposition which will be useful in bounding the expected cost of the algorithm.

\begin{proposition}
    \label{prop:alg}
    There exists a universal constant $\kappa > 0$ such that for every $M\geq 2$, we have
    \[\sum_{t=1}^{\lfloor M/2\rfloor}\frac{\lfloor M/2\rfloor-t+1}{M-t+1} \geq \kappa \cdot M.\]
\end{proposition}

\begin{proof}
We have
        \begin{align*}
            \sum_{t=1}^{\lfloor M/2\rfloor}\frac{\lfloor M/2\rfloor-t+1}{M-t+1}
            &=\sum_{t=1}^{\lfloor M/2\rfloor}\frac{(M-t+1)-\lceil M/2\rceil }{M-t+1}\\
        &=\lfloor M/2\rfloor-\lfloor M/2\rfloor\sum_{i=\lceil M/2\rceil+1}^{M}\frac{1}{i}\\
        &= \lfloor M/2\rfloor-\lceil M/2\rceil(H_{M}-H_{\lceil M/2\rceil})\\
        &=\lfloor M/2\rfloor-\lceil M/2\rceil(\ln2+o(1))=\frac{1-\ln 2}{2}M(1+o(1)),\end{align*}
    where $H_{i}$ is the $i$th harmonic number.
\end{proof}

\begin{theorem}\label{level1}
    The expected competitive ratio of every deterministic algorithm against the uniform random permutation of $m$ request types on $\mathcal{T}_{1,m}$ is $\Omega(m)$.
\end{theorem}

\begin{proof}
    First, we show that there is a solution with cost at most $2$.
    The whole sequence $\rho_{1}, \ldots, \rho_{m}$ has each of the possible request types $K_{1},\ldots,K_{m}$ exactly once. The solution processes requests $\rho_{1},\ldots,\rho_{\lceil m/2\rceil }$ at the state labeled by the set of all indices of request types in $\rho_{1},\ldots,\rho_{\lceil m/2\rceil}$, i.e. at $S=\{i\mid K_{i}\text{ is in subsequence }\rho_{1},\ldots,\rho_{\lceil m/2\rceil}\}$, and then moves to the state labeled by the set of all indices of request types in the subsequence $\rho_{\lfloor m/2 \rfloor +1},\ldots, \rho_{m}$. The total cost of this solution is at most $2$, no matter which initial state the server is in.

We now lower bound the expected cost of every deterministic algorithm. Fix a deterministic algorithm. For each step $t$, let $R_{t}$ be the set of request types that have not appeared prior to this step and $S_{t-1}$ denote the label of the state of the algorithm right before the arrival of $\rho_{t}$. It holds that $|R_{t}|=m-t+1$ and $|S_{t-1}|=\lceil m/2\rceil $, and hence $|R_{t}\setminus S_{t-1}|\geq \lfloor m/2\rfloor -t+1$. Let $Y_{t}$ denote the event that the algorithm has to move at step $t$ (i.e.~because $\rho_{t}\notin S_{t-1}$). Since $\rho_t$ is chosen u.a.r.~from $R_t$, we have $\mathbb{E}[Y_{t}]\geq (\lfloor m/2\rfloor -t+1)/(m-t+1)$ for every $1\leq t\leq \lfloor m/2\rfloor$. For $\lfloor m/2\rfloor+1\leq t\leq m$, we use the trivial bound $\mathbb{E}[Y_{t}]\geq 0$. Thus, the total expected cost of the deterministic algorithm is 
    \begin{align*}       
    \sum_{t=1}^{m}\mathbb{E}[Y_{t}]&\geq \sum_{t=1}^{\lfloor m/2\rfloor}\frac{\lfloor m/2\rfloor-t+1}{m-t+1}=\Omega(m),\end{align*}
    by \cref{prop:alg}.
\end{proof}

\subsubsection{Lifting to Higher Levels}
\label{lb-lift}
For every $k$ and $m$, we construct an HST $\mathcal{T}_{k,m}$ to lower bound $c_{k,m}$.

\subparagraph{Key property of $\cT_{k,m}$.}
The key property of $\cT_{1,m}$ and the distribution of request sequences is that for every request sequence in the support of the distribution, there is a leaf that can serve the first half of the requests, and another leaf that can serve the second half. Thus, there is a solution that only moves twice. We generalize this to $\cT_{k,m}$ as follows. The distribution of request sequences and $\cT_{k,m}$ have the property that for every sequence in the support of the distribution, there are two ``safe" $(k-1)$-level subtrees such that there is a low-cost solution that can serve the first half of the request sequence in one of them and the second half in the other. In other words, the low-cost solution only moves twice between $(k-1)$-level subtrees.

\subparagraph{Higher-level labels.} 
In our construction of $\cT_{k,m}$, we will assign labels to all nodes except the root. The labels for leaves come from the set
$\mathcal{P}_{0}:=\binom{[m]}{\lceil m/2\rceil}$\footnote{The notation $\binom{S}{k}$ means the collection of all $k$-size subsets of $S$.}. The labels for $\ell$-level nodes are collections of those for $(\ell-1)$-level nodes as follows: Let $\mathcal{P}_{\ell}$ be the set of $\ell$-level labels, then
$\mathcal{P}_{\ell}:=\binom{\mathcal{P}_{\ell-1}}{\lceil|\mathcal{P}_{\ell-1}/2|\rceil }$ for every $1\leq \ell\leq k-1$. It is easy to see that $|\mathcal{P}_{\ell}|=A_{\ell+1,m}$. We also define the label of a subtree to be the label of its root node.

\subparagraph{Construction of $\cT_{k,m}$.}
For $k \geq 2$, the HST $\cT_{k,m}$ is constructed iteratively in a top-down fashion as follows. For every label $P \in \cP_{k-1}$, the root of the HST has a child labeled $P$. For every level $k-1\geq \ell \geq 1$, every level-$\ell$ node $v$ has $|P(v)|$ children where $P(v)$ is its label, each given a unique label $p \in P(v)$.\footnote{Recall that $P(v)$, a label of a level-$\ell$ node, is a collection of level-($\ell-1$) labels.} We remark that there is exactly one level-$(k-1)$ node per label in $\cP_{k-1}$ but for $\ell < k-1$, there are several level-$\ell$ nodes from different $(\ell+1)$-level subtrees that share the same label. For every $1\leq \ell\leq k$, the distance between every two leaves whose lowest common ancestor is on level $\ell$ is $C_{\ell-1}$\footnote{We always assume that $C_\ell/C_{\ell-1}$ is an integer for every $\ell\geq 1$ .} ($C_{0}=1$), where $C_{\ell}$'s depend on $k$ and $m$ and will be defined later in \cref{sec:generalhst}.

\begin{observation}
    In $\mathcal{T}_{k,m}$, every $\ell$-level node has exactly $\left\lceil A_{\ell,m}/2\right\rceil$ level-$(\ell-1)$ children with distinct labels, for every $1\leq \ell \leq k-1$.
\end{observation}

\subsubsection{Lower Bound Sequence for $\cT_{2,m}$}
\label{lb-2}

We now prove our lower bound of $c_{2,m}$. Recall that  $C_1$ is the distance between two leaves in different $1$-level subtrees, and $\cP_0$ is the set of all $\lceil m/2\rceil$-size subsets of $[m]$.
 For each label $S \in \cP_0$, define  $\brho(S)$ to be an arbitrary permutation of $S$.

Now we describe the adversarial randomized input sequence. The adversary first samples a uniformly random permutation $\{S_{t}\}_{t=1}^{|\cP_0|}$ of $\mathcal{P}_{0}$ and then proceeds in $|\cP_0|=A_{1,m}$ phases as follows: in phase $t$, it repeats the sequence $\brho(S_{t})$ for $C_{1}$ times. Thus,  the entire input is a  concatenation of all those phases:
\[
\underbrace{\brho(S_1), \ldots, \brho(S_1)}_{\text{$C_1$ times}},
\underbrace{\brho(S_2), \ldots, \brho(S_2)}_{\text{$C_1$ times}},\ldots,
\underbrace{\brho(S_{|\cP_0|}), \ldots, \brho(S_{|\cP_0|})}_{\text{$C_1$ times}}.
\]

For every $1\leq t\leq A_{1,m}$, we use $T^{1}_{t-1}$ to denote the label of the $1$-level subtree in which the state of the algorithm is located right before the beginning of the phase $t$.

\begin{lemma}\label{level2}
For every $C_1$, every deterministic algorithm has expected cost of at least $\Omega(A_{1,m}C_{1})$.
\end{lemma}

\begin{proof}
    We first show that: the cost incurred during the whole phase $t$ is at least $C_{1}$ if $T^{1}_{t-1}$ does not contain $S_{t}$. If $S_{t}\notin T^{1}_{t-1}$, either the algorithm moves to another $1$-level subtree that contains $S_{t}$ as a leaf node, or the algorithm needs to move at least once during every loop of $\brho(S_{t})$ since there is no state in $T^{1}_{t-1}$ that can satisfy all request types in $S_{t}$. In either way, the cost is at least $C_{1}$.

Now we consider the probability of the event that $T^{1}_{t-1}$ does not contain $S_{t}$. Similar to the proof of Theorem~\ref{level1}, at every phase $1\leq t\leq \left\lfloor {A_{1,m}}/{2}\right\rfloor$, the probability that $T^{1}_{t-1}$ does not contain a node labeled $S_{t}$ is at least ${(\lfloor A_{1,m}/2 \rfloor-t+1)}/{ (A_{1,m}-t+1)}$ because $T^{1}_{t-1}$ contains exactly $\left\lceil {A_{1,m}}/{2}\right\rceil$ different leaf labels. So, \cref{prop:alg} implies that  the total expected cost incurred during the first $\left\lfloor {A_{1,m}}/{2}\right\rfloor$ phases is $\Omega(A_{1,m}C_{1})$.
\end{proof}

\begin{lemma}
    For every $C_1$, the optimal offline cost is at most $2C_{1}+A_{1,m}$ no matter where the initial state is.
\end{lemma}
\begin{proof}
    Consider the following solution. It moves to the subtree labeled $\{S_{1},\ldots,S_{\left\lceil A_{1,m}/2\right\rceil}\}$ and for phases $1 \leq t \leq \left\lfloor A_{1,m}/2\right\rfloor$, it serves the repetitions of the subsequence $\brho(S_t)$ at the node labeled $S_t$ in the subtree. Then, it moves to the subtree labeled $\{S_{\lfloor A_{1,m}/2\rfloor+1},\ldots,S_{A_{1,m}}\}$ and for $\left\lfloor A_{1,m}/2\right\rfloor < t \leq A_{1,m}$, serves the repetitions of the subsequence $\brho(S_t)$ at the node labeled $S_t$ in the subtree. It incurs a cost of at most $2C_1$ for movement between different $1$-level subtrees and at most $A_{1,m}$ for movement within each of the two $1$-level subtrees. Thus, the total cost is at most $2C_1 + A_{1,m}$.  
\end{proof}

Since the two lemmata that we have just proven hold for every value $C_{1}$, we can set it to be some value such that $C_{1}=\omega (A_{1,m})$ with both values regarded as functions of $m$, and the desired lower bound immediately follows.

\begin{theorem}\label{lower2}
    Every randomized algorithm has a worst case competitive ratio $\Omega(A_{1,m})$ on some instances of $2$-level HSTs.
\end{theorem}

\subsubsection{General Iterative Lower Bounds on Higher Levels}
\label{sec:generalhst}
We now prove the lower bound on $c_{k,m}$ for $k \geq 3$ expressed with the iterative sequence $\{A_{k,m}\}_{k\geq 1}$ defined in \cref{eq:iter} at the beginning of \cref{iterativesection}.

\begin{theorem}
\phantomsection\label{lowerboundG}
    $c_{k,m}=\Omega (A_{k-1,m})$ for every $k\geq 3$.
\end{theorem}
\noindent \cref{imposs} now follows from \cref{lowerboundG} and \cref{iterdiv}.

\subparagraph{Sequences of labels.} 

We have provided the construction of $\mathcal{T}_{k,m}$ for every $k$, and now we need to describe the adversary's strategy to generate a randomized input sequence on which no deterministic algorithm achieves competitive ratio better than $c_{k,m}$. Recall that every sequence in the support of the randomized input on $\mathcal{T}_{2,m}$ is a concatenation of sequences $\brho(S)$ for labels $S\in \mathcal{P}_{0}$. Now we want to generalize this concept, the sequence of a label, to every level-$\ell$ label (i.e. every element of $\mathcal{P}_{l}$).

 We define inductively on $\ell$ with $\ell = 0$ as the base case. For every $\ell\geq 1$ and every label $T^{\ell}\in \mathcal{P}_{\ell}$, the sequence $\brho(T^{\ell})$ is constructed in $|T^{\ell}|$ phases as follows. Let $T^{\ell-1}_1,\ldots,T^{\ell-1}_{|T^\ell|}$ be an arbitrary permutation of $T^\ell$. During each phase $1 \leq t \leq |T^\ell|$, the sequence $\brho(T^{\ell-1})$ is repeated consecutively for $\tfrac{C_{\ell}}{C_{\ell-1}}$ (we can choose values of $C_{\ell}$'s to make this amount always an integer at the end) times, and the entire sequence $\brho(T^{\ell})$ is a concatenation of those phases:
 \[
\underbrace{\brho(T^{\ell-1}_{1}), \ldots, \brho(T^{\ell-1}_{1})}_{\text{$\frac{C_{\ell}}{C_{\ell-1}}$ times}},
\underbrace{\brho(T^{\ell-1}_{2}), \ldots, \brho(T^{\ell-1}_{2})}_{\text{$\frac{C_{\ell}}{C_{\ell-1}}$ times}},\ldots,
\underbrace{\brho(T^{\ell-1}_{|T^{\ell}|}), \ldots, \brho(T^{\ell-1}_{|T^{\ell}|})}_{\text{$\frac{C_{\ell}}{C_{\ell-1}}$ times}}.
\]

\begin{lemma}\label{atleast}
Let $1\leq \ell\leq k-2$ and $P\in\mathcal{P}_{\ell}$. Consider a solution for $\brho(P)$ whose initial state is not in an $\ell$-level subtree with label $P$. Then, the solution incurs a cost of at least $C_\ell$ on $\brho(P)$.
\end{lemma}

\begin{proof}
    We prove by induction on $\ell$. The base case is when $\ell=1$. Fix some $P\in\mathcal{P}_{1}$. The sequence $\brho(P)$ consists of $C_{1}$ consecutive repetitions of the sequence $\brho(p)$ for every child $p\in P$. If the initial state is in a $1$-level subtree labeled ${T}^{1}\neq P$, then there exists some $p\in P$ such that $p \notin T^{1}$ and thus there is no leaf in the subtree labeled $p$. As a result, the algorithm either visits a different $1$-level subtree, or incurs cost $1$ during every repetition of $\brho(p)$, thus costing at least $C_{1}$.

    Suppose the claim holds for every sequence $\brho(p)$ such that $p\in\mathcal{P}_{\ell}$ ($1\leq \ell\leq k-3$). Now consider a sequence $\brho(P)$ for some $P\in\mathcal{P}_{\ell+1}$. For every $p \in P$, it contains $\frac{C_{\ell+1}}{C_{\ell}}$ consecutive repetitions of  $\brho(p)$. If the initial state is in an $(\ell+1)$-level subtree labeled $T^{\ell+1}\neq  P$, then there exists $p\in P\setminus T^{\ell+1}$ because $|T^{\ell+1}|=|P|$. 
If the algorithm moves to a different $(\ell+1)$-level subtree, it incurs a cost of $C_{\ell+1}$.     
Now suppose the algorithm does not move to a different $(\ell+1)$-level subtree. Since $p \notin T_{\ell+1}$, the current subtree does not have an $\ell$-level subtree labeled $p$. Thus at the start of each of the $C_{\ell+1}/C_\ell$ repetitions of $\brho(p)$, by the inductive hypothesis, it incurs a cost of at least $C_{\ell}$. The total cost is thus at least $C_{\ell+1}$. 
\end{proof}

We now define another iterative sequence $\{B_{\ell}\}_{1\leq\ell \leq k-1}$ and show that it can be used to express the upper bound of the optimal solution of the request sequences:
\begin{align*}
    B_{1}=\left\lceil \frac{A_{1,m}}{2}\right\rceil, \qquad
    B_{\ell}=\left\lceil\frac{A_{\ell,m}}{2}\right\rceil \left[(B_{\ell-1}+C_{\ell-2})\cdot\frac{C_{\ell}}{C_{\ell-1}}+C_{\ell-1}\right]  \quad \mbox{for $2\leq \ell\leq k-1$.}   
\end{align*}

\begin{lemma}\label{Bk}
    For every label $T^{\ell}$ in $\mathcal{P}_{\ell}$ and subtree labeled $T^\ell$ in $\mathcal{T}_{k,m}$, there exists a trajectory contained within the subtree that incurs cost of no more than $B_{\ell}$ to serve the sequence $\brho(T^{\ell})$, for every $1\leq \ell\leq k-1$.
\end{lemma}

\begin{proof}
We prove by induction on $\ell$. Consider the base case $\ell=1$.
    Fix some $T^{1}\in\mathcal{P}_{1}$. Order the elements $S_1, \ldots, S_{\left\lceil A_{1,m}/2\right\rceil}$ of $T$ in the same order as the phases in $\brho(T^{1})$. Consider the trajectory that first moves to the leaf node labeled $S_{1}$ in the $1$-level subtree labeled $T^{1}$ and then  moves to leaves $S_{2},\ldots,S_{\left\lceil A_{1,m}/2\right\rceil}$ one by one with cost $1$ each time since they are all in the same $1$-level subtree. The total cost is at most $\left\lceil\frac{A_{1,m}}{2}\right\rceil=B_{1}$.
    
    Now we consider $\brho(T^{\ell})$ for $2\leq \ell\leq k-1$, given that the property for $\ell-1$ has been proven. Suppose that $\brho(T^{\ell})$ is the concatenation of $\brho(T^{\ell-1}_{1}),\ldots,\brho(T^{\ell-1}_{\left\lceil A_{\ell,m}/2 \right\rceil})$ by the chronological order where each $\brho(T^{\ell-1}_{j})$ is repeated for $\frac{C_{\ell}}{C_{\ell-1}}$ times before proceeding to the next one. 
    Consider the trajectory that for each phase $1 \leq j \leq \lceil A_{\ell,m}/2\rceil$, consists of $\frac{C_\ell}{C_{\ell-1}}$ copies of the trajectory obtained by applying the inductive hypothesis to $\brho(T^{\ell-1}_j)$. In each phase $j$, for each copy of $\brho(T^{\ell-1}_j)$, the trajectory incurs a cost of at most $C_{\ell-2}$ to move to its initial state within the subtree with label $T^{\ell-1}_{j}$ and then a cost of $B_{\ell-1}$ to serve $\brho(T^{\ell-1}_j)$ within the same subtree. Thus, for each phase $j$, the trajectory's cost is at most $\frac{C_{\ell}}{C_{\ell-1}}(B_{\ell-1}+C_{\ell-2})$ plus at most $C_{\ell-1}$ to move to an $(\ell-1)$-level subtree labeled $T^{\ell-1}_j$ at the start of the phase. As there are $\left\lceil{A_{\ell,m}}/{2}\right\rceil$ phases, the total cost of the trajectory is bounded by $B_\ell$, as desired.
\end{proof}

\subparagraph{Adversary's strategy and analysis} We are finally ready to describe the adversary's strategy on $\mathcal{T}_{k,m}$. It first  samples a uniformly random permutation  $\{T^{k-2}_{t}\}_{t=1}^{A_{k-1,m}}$ of $\mathcal{P}_{k-2}$ (note that $|\mathcal{P}_{k-2}|=A_{k-1,m}$), and then proceeds in $A_{k-1,m}$ phases: in each phase $t$, $\brho(T^{k-2}_{t})$ is repeated consecutively for $\frac{C_{k-1}}{C_{k-2}}$ times, and the whole input is the concatenation of all $A_{k-1,m}$ phases.

\begin{lemma}
    The optimal cost is no more than $2(C_{k-1}+B_{k-1})$ regardless of the initial state.
\end{lemma}
\begin{proof}
The whole sequence can be regarded as the concatenation of the sequences $\brho(T^{k-1})$'s of two $(k-1)$-level subtrees: 
\begin{equation*}
    \{T^{k-2}_{1},\ldots,T^{k-2}_{\lceil A_{k-1,m}/2\rceil }\}\text{ and } \{T^{k-2}_{\lfloor A_{k-1,m}/2\rfloor+1 },\ldots,T^{k-2}_{ A_{k-1,m} }\}.
\end{equation*}
The repetitions of $\brho(T^{k-2}_{\lfloor A_{k-1,m}/2\rfloor+1 })$ appear in both of the $(k-1)$-level subtrees' sequence when $m$ is odd, but it does not affect our conclusion since we are proving an upper bound. As a result, the optimal solution just pays the cost at most $B_{k-1}$ in each $T^{k-1}$ (Lemma~\ref{Bk}) plus the cost at most $2C_{k-1}$ to visit two $(k-1)$-level subtrees.
\end{proof}

\begin{lemma}
    The expected cost of every deterministic algorithm to serve this randomized input on $\mathcal{T}_{k,m}$ is $\Omega(A_{k-1,m}C_{k-1})$.
\end{lemma}
\begin{proof}
The proof is similar those of
\cref{level1} and \cref{level2}. Fix a deterministic algorithm.
During every phase $1\leq t\leq \left\lfloor {A_{k-1,m}}/{2}\right\rfloor$ when some $\brho(T^{k-2}_{t})$ is repeated for $\frac{C_{k-1}}{C_{k-2}}$ times, the probability that  $\brho(T^{k-2}_{t})$ is not in the current $(k-1)$-level subtree containing the algorithm's server has  probability is at least $\frac{\lfloor A_{k-1,m}/2\rfloor -t+1}{A_{k-1,m}-t+1}$. When this happens, the algorithm only has two choices: either move to a new $(k-1)$-level subtree with cost $C_{k-1}$, or incur cost of at least $C_{k-2}$ during each repetition of $\brho(T^{k-2}_{t})$ according to Lemma~\ref{atleast} since the initial state cannot be in a $(k-2)$-level subtree with label $T^{k-2}_{t}$.
\end{proof}

It remains to show that there exist possible values of $C_{\ell}$'s and $B_{\ell}$'s for $1\leq\ell\leq k-1$ such that $C_{k-1}=\omega(B_{k-1})$ as $m$ tends to infinity\footnote{Here both $C_{k-1}$ and $B_{k-1}$ are functions of $m$. When we consider both $m$ and $k$ as constants, the lower bound can be achieved by setting the ratio $C_{\ell}/C_{\ell-1}$ sufficiently large for every $1\leq\ell\leq k-1$ and the proof idea is similar.}. This would show that the randomized competitive ratio on $\mathcal{T}_{k,m}$ is at least $\Omega(\frac{A_{k-1,m}C_{k-1}}{C_{k-1}+B_{k-1}})=\Omega(A_{k-1,m})$ as a function of $m$.

Let $C_{1}=\omega(\prod_{j=1}^{k-1}A_{j,m})$ and $C_{\ell}=\omega(  (\prod_{j=\ell}^{k-1}A_{j,m})C_{\ell-1})$ for $2\leq \ell\leq k-1$. We now prove that such a sequence satisfies the desired property.

\begin{lemma}
    $C_{k-1}=\omega(B_{k-1})$.
\end{lemma}

\begin{proof}
    We prove that $C_{\ell}=\omega( (\prod_{j=\ell+1}^{k-1}A_{j,m}) B_{\ell})$ as $m$ tends to infinity for every $1\leq \ell\leq k-1$ inductively, where $\prod_{j=k}^{k-1}A_{j,m}$ is defined to be $1$.

    Clearly the result holds for $\ell=1$ because $B_{1}=\Theta (A_{1,m})$. With the inductive hypothesis that $C_{\ell}=\omega(  (\prod_{j=\ell+1}^{k-1}A_{j,m}) B_{\ell})$ for some $1\leq \ell\leq k-2$, then in order to prove that $C_{\ell+1}=\omega( (\prod_{j=\ell+2}^{k-1}A_{j,m}) B_{\ell+1})$, according to the definition that $B_{\ell+1}=\left\lceil\frac{A_{\ell+1,m}}{2}\right\rceil [(B_{\ell}+C_{\ell-1})\cdot\frac{C_{\ell+1}}{C_{\ell}}+C_{\ell}]$, it suffices to show that the following two statements both hold:
    \begin{equation}
        C_{\ell+1}=\omega( A_{\ell+1,m}\cdot (\prod_{j=\ell+2}^{k-1}A_{j,m})C_{\ell})=\omega(\prod_{j=\ell+1}^{k-1}A_{j,m}C_{\ell}),
    \end{equation}
    which is exactly our definition of $C_{\ell+1}$, and
    \begin{equation}\label{6}
        C_{\ell}=\omega( (\prod_{j=\ell+1}^{k-1}A_{j,m})(B_{\ell}+C_{\ell-1})),
    \end{equation}
while (\ref{6}) can be further reduced to that the following two formulae $C_{\ell}=\omega( (\prod_{j=\ell+1}^{k-1}A_{j,m})B_{\ell})$ and $C_{\ell}=\omega( (\prod_{j=\ell+1}^{k-1}A_{j,m})C_{\ell-1})$ both hold, where the first one holds due to the inductive hypothesis and the second holds due to the definition of $C_{\ell}$.
\end{proof}

\noindent This completes the proof of Theorem~\ref{lowerboundG}.

\begin{remark}
    Our iterative sequence of lower bounds on $k$-level HSTs marks an improvement from the one provided in Section 4.2 of~\cite{DBLP:conf/approx/BubeckR20} in the sense that: given every fixed $k\geq 2$, the lower bound provided by our sequence is asymptotically tighter than theirs when the parameter $m$ tends to infinity. To formalize, consider their iterative sequence $\{\hat{c}_{k,m}\}_{k\geq 1}$ defined as the following: $\hat{c}_{1,m}=\left\lfloor m/2\right\rfloor $ and $\hat{c}_{k,m}=2^{\hat{c}_{k-1,m}-1}$. We can use inductive arguments on $k$ to show that, $A_{k,m}=\omega(\hat{c}_{k,m})$ for every $k\geq 2$: it is clear that the statement holds for $k=2$, since $A_{2,m}=\binom{m}{\lceil m/2\rceil}=\Theta (2^{m}/ \sqrt{m})$ and $\hat{c}_{2,m}=\Theta(2^{m/2})$; suppose that $\lim_{m\to \infty}A_{k-1,m}/\hat{c}_{k-1,m}=\infty$, then we show that $\lim_{m\to \infty}\log({A_{k,m}}/{\hat{c}_{k,m}})=\infty$ and hence $\lim_{m\to \infty}A_{k,m}/{\hat{c}_{k,m}}=\infty$. To see this, $\log({A_{k,m}}/{\hat{c}_{k,m}})=\log A_{k,m}-\log \hat{c}_{k,m}=A_{k-1,m}-{\log (A_{k-1,m}})/{2}-\hat{c}_{k-1,m}+O(1)$. Due to the inductive hypothesis, $\hat{c}_{k-1,m}=o(A_{k-1,m})$. Therefore, $A_{k-1,m}-{\log (A_{k-1,m}})/{2}-\hat{c}_{k-1,m}\sim A_{k-1,m}$ and $\lim_{m\to \infty}\log({A_{k,m}}/{\hat{c}_{k,m}})=\infty$.
    
\end{remark}

\subsection{Tight Upper Bound for 2-Level HSTs}
\label{2-tight}
The idea to show a general upper bound for $2$-level HSTs is as follows: even if there is no restriction on the size of the whole tree itself, some leaves and subtrees are redundant in the sense that any reasonable solution should not move to them, or does not need to move to them while maintaining feasibility and not increasing the cost. As a result, we can trim off those redundant nodes and hence reduce the graph size, just as how we show that we only need to consider $2^{m}-1$ states for the weighted star case. This idea has also been applied by Bubeck and Rabani to give iterative upper bounds for $2$ and higher levels of HSTs (see Section 4.2 Deeper Trees in~\cite{DBLP:conf/approx/BubeckR20}), and we present a refined argument on $2$-level HSTs that closes the gap between upper and lower bounds. We show that the number of leaves remained after such preprocessing is bounded in terms of $m$ and apply the optimal MTS algorithm for HSTs from~\cite{DBLP:journals/siamcomp/BubeckCLL21}.

Since the initial state can be arbitrarily given, in this part we only consider those leaves different from the initial state, to evaluate whether some states can be totally replaced by others and every algorithm can still maintain a feasible solution with equal or less cost. This is without loss of generality up to a constant factor, since there is only one initial state and thus it is sufficient to give an upper bound of the amount of states that are not the initial state.

In this part, we express leaf nodes lexicographically by a tuple $(S,T)$ where $S$ is the label of the leaf node itself and $T$ is the label of the $1$-level subtree that it is in. The label of every node is defined in the same way as the previous subsection\footnote{Recall that the label of a leaf node is the set of request types that this state can satisfy, and the label of a $1$-level subtree is the set of labels of leaf nodes contained in it.} . A preliminary observation is that, if two leaf nodes share exactly the same tuple, then they are isomorphic and it suffices to only keep one in the tree. The feasibility of removing redundant nodes is that, when $m$ possible requests and their corresponding feasible sets in the whole metric are fixed, running an algorithm over a subset of states produces a feasible solution for the whole metric as well, therefore what we need to show is that the optimal cost does not increase while we trim off some states.

\begin{proposition}
    If there are two leaves $(S_{1},T)$ and $(S_{2},T)$ in the same $1$-level subtree such that $S_{1}\subseteq S_{2}$, then any move to $(S_{1}, T)$ can be replaced with one to $(S_{2},T)$ without increasing the cost.
\end{proposition}

\begin{proof}
    Any request that is feasible at $(S_{1}, T)$ is also feasible at $(S_{2}, T)$ since $S_{1}\subseteq S_{2}$. Moreover, for any state different from $(S_{1}, T)$ and $(S_{2}, T)$, it has the same distance from those two nodes that are within the same $1$-level subtree $T$ according to the $2$-level HST structure.
\end{proof}

As a result, we can assume that every $1$-level subtree is an antichain over subsets of $[m]$, which contains at most $\binom{m}{\lceil m/2\rceil }=A_{1,m}$ elements according to Sperner's theorem~\cite{sperner1928satz}. The next thing is to consider the maximum amount of such $1$-level subtrees, which is equivalent to the amount of antichains over subsets of $[m]$. This amount is called Dedekind number $D(m)$ and its following asymptotic approximation property has been proven in~\cite{kleitman1975dedekind}:

\begin{lemma}[\cite{kleitman1975dedekind}]\label{dedekind}
    $\log D(m)=O(A_{1,m})$.
\end{lemma}

\begin{theorem}\label{up2}
    $c_{2,m}=O(A_{1,m})$.
\end{theorem}

\begin{proof}
    It suffices to consider at most $O(D(m)A_{1,m})$ leaf nodes (including the initial state and at most $A_{1,m}$ other leaves in the same $1$-level subtree as they cannot be pruned by the above process) and apply the optimal online MTS algorithm on an HST that has competitive ratio $O(\log n)$, where $n$ is the number of points in the metric.  According to the asymptotic property of $D(m)$ as shown in Lemma~\ref{dedekind}, $\log (D(m)A_{1,m})=O(A_{1,m}+\log A_{1,m})=O(A_{1,m})$.
\end{proof}

%% file: two-sets.tex
\section{Chasing Two Sets on General Metrics}

In this section, we consider the case $m=2$, where there are only two possible request types $K_{1}$ and $K_{2}$ that can be generated by the adversary, and show that there is a constant-competitive algorithm that works on general metrics, with no requirement on the size $|X|$ nor the structure of the metric space.

Our method is an application of the Guess-and-Double Method. Here we address that our algorithm is strictly competitive without an additive constant. This matters when $m=2$ because it is trivial to design a non-strictly $1$-competitive algorithm in this case: suppose $s_{1}\in K_{1}$ and $s_{2}\in K_{2}$ are the points such that $d(s_{1},s_{2})=d(K_{1},K_{2})$, then the algorithm can move to $s_{1}$ upon the first request of $K_{1}$ and then always alternate between $s_{1}$ and $s_{2}$ upon every new request, and it is easy to see that the cost of this algorithm is no more than the optimal cost plus the diameter of the metric space. As a result, we can apply the Guess-and-Double Method to eliminate the additive constant and devise a strictly competitive algorithm.

\subsection{Counterexample to Greedy}

Before we show how to apply the Guess-and-Double Method to our problem setting, we first use a simple example to see how the classical ski rental problem can be reduced to our problem, which denies the attempt of applying the greedy algorithm that always moves to the closest feasible state from the current one.

\begin{example}
Suppose that we have a ski rental instance where the rental rate is $1$ and the buying cost is $B>1$. At every new step of the online instance, if the algorithm has not bought the ski set yet, then it needs to either rent it for one more step with cost $1$, or buy it with cost $B$. After buying the ski set, the algorithm does not need to pay new cost no matter how many steps are there in the future.

Consider a metric space where there are four points $s_{0},s_{1},s_{2},s_{3}$ such that $d(s_{0},s_{1})=d(s_{1},s_{2})=d(s_{0},s_{2})=1$ and $d(s_{i},s_{3})=B$ for $i=0,1,2$. We construct two sets $K_{1}$ and $K_{2}$ such that $K_{1}=\{s_{1},s_{3}\}$ and $K_{2}=\{s_{2},s_{3}\}$. The initial location of the server is at $s_{0}$. $s_{3}$ here is the safe location for both $K_{1}$ and $K_{2}$, which corresponds to the state of the ski rental instance after buying the ski set with cost $B$.

The input sequence of the two-set chasing problem instance that we construct is just $K_{1}$ and $K_{2}$ alternating. It is easy to see that if the server has not moved to the safe location $s_{3}$ with cost $B$ at some step, then every time when a new request appears, it always needs to switch between $s_{1}$ and $s_{2}$ with cost $1$. This is exactly equivalent to the two-state ski rental problem. Moreover, we can also see from this reduction that the simple memoryless greedy algorithm to move to the closest feasible state at every step has unbounded competitive ratio when the input sequence is sufficiently long, since it pays cost $1$ at every step and never visits $s_{3}$.
\end{example}

\subsection{Our Algorithm}
Without loss of generality, suppose that the first request in the sequence that forces the algorithm to move is $K_{1}$ and hence the whole sequence can be regarded as $K_{1},K_{2},K_{1},K_{2}\ldots$, strictly alternating between two sets. In this way, $s_{0}\notin K_{1}$ holds. We further assume that $d(s_{0},K_{1})=\min_{s\in K_{1}}d(s_{0},s)= 1$, which means that the cost of the optimal solution is always at least $1$ as long as the input sequence is non-empty. Moreover, if there is only one time step in the input sequence, then the optimal solution is exactly $1$.

Since the sequence is essentially in the alternating form $K_{1},K_{2},K_{1},K_{2}\ldots$, the only variable we have to deal with is the length of the sequence. Longer the sequence, bigger the value of the optimal solution, and at every step we have a lower bound of the optimal solution which is exactly that of the partial sequence that has appeared so far. To formalize, Let $\OPT_{t}$ be the cost of the optimal solution for the server to process the tasks that have arrived up to and including step $t$. At every step $t$ we maintain a guess of an integer $j\geq 0$ such that

\begin{equation}
    2^{j}\leq \OPT_{t}<2^{j+1}.
\end{equation}

Moreover, we can compute the input sequence with the largest number of steps whose optimal value is in the range $[2^{j}, 2^{j+1})$ for each non-negative integer $j$, and let $\widehat{\OPT}_{j}$ be the optimal solution of serving this sequence. We slightly abuse similar notations by using them to refer to both the solution itself and the cost value of it. Note that if $2^{j}\leq \OPT_{t}<2^{j+1}$ for some $t$, then $s_{t}(\widehat{\OPT}_{j})$, the state of $\widehat{\OPT}_{j}$ at step $t$, is well-defined, because the length of the sequence for which $\widehat{\OPT}_{j}$ is defined is at least $t$ based on our choice of the longest possible one. Our algorithm always moves to $s_{t}(\widehat{\OPT}_{j})$ as we have defined at every step $t$.

\begin{algorithm}
\caption{A deterministic algorithm for the two-set chasing problem based on the Guess-and-Double Method}\label{alg2}

\textbf{Initialize:} $j=0$ since $\OPT_{1}=1$.

At every time step $t$:

\begin{algorithmic}
    \If {$2^{j}\leq\OPT_{t}<2^{j+1}$}
    \State Move to state $s_{t}(\widehat{\OPT}_{j})$, the state of $\widehat{\OPT}_{j}$ at step $t$.
\Else

End the previous phase;

Update the value $j$ such that $2^{j}\leq\OPT_{t}<2^{j+1}$ for the new phase;

Move to state $s_{t}(\widehat{\OPT}_{j})$.
\EndIf
\end{algorithmic}

\end{algorithm}

We divide the whole process into phases according to the values $j$'s that are used in the Guess-and-Double Method where each phase contains all steps when the value $j$ stays a given value (see the description of Algorithm~\ref{alg2} for details). Suppose that there are $k$ phases in total chronologically with values $j_{1},...,j_{k}$ respectively, then the final optimal solution has value $2^{j_{k}}\leq\OPT<2^{j_{k}+1}$, according to the definition of the guessed value $j$. When the algorithm moves from state $s_{t}(\widehat{\OPT}_{j_{\ell}})$ to $s_{t+1}(\widehat{\OPT}_{j_{\ell}})$ within the same phase with the guessed value $j_{\ell}$, we charge the cost $d(s_{t}(\widehat{\OPT}_{j_{\ell}}), s_{t+1}(\widehat{\OPT}_{j_{\ell}}))$ to the \emph{internal cost} of the $l$'th phase; when it moves from $s_{t}(\widehat{\OPT}_{j_{\ell}})$ to $s_{t+1}(\widehat{\OPT}_{j_{\ell+1}})$ between two different phases, we first use the triangle inequality to give an upper bound
\begin{equation*}
   d(s_{t}(\widehat{\OPT}_{j_{\ell}}), s_{t+1}(\widehat{\OPT}_{j_{\ell+1}}))\leq d(s_{t}(\widehat{\OPT}_{j_{\ell}}), s_{0})+d(s_{0}, s_{t+1}(\widehat{\OPT}_{j_{\ell+1}})), 
\end{equation*}
and then charge $d(s_{0}, s_{t+1}(\widehat{\OPT}_{j_{\ell+1}}))$ to the internal cost of the $(\ell+1)$'th phase and  $d(s_{t}(\widehat{\OPT}_{j_{\ell}}), s_{0})$ to the \emph{external cost} between the $\ell$'th and the $(\ell+1)$'th phase.

\begin{lemma}
    The internal cost of the $\ell$'th phase is no more than $2^{j_{\ell}+1}$ for every $1\leq \ell\leq k$.
\end{lemma}

\begin{proof}
    Suppose that the Guess-and-Double Method stays in the $\ell$'th phase from step $t_{\ell,1}$ to $t_{\ell,2}$, then the internal cost of this phase can be expressed as 
    \begin{equation*}
            d(s_{0}, s_{t_{\ell,1}}(\widehat{\OPT}_{j_{\ell}}))+\sum_{t=t_{\ell,1}}^{t_{\ell,2}-1}d(s_{t}(\widehat{\OPT}_{j_{\ell}}),s_{t+1}(\widehat{\OPT}_{j_{\ell}})).
    \end{equation*}

Based on our notation of the trajectory of $\widehat{\OPT}_{j_{\ell}}$ and the definition that it is the longest sequence whose optimal value is strictly less than $2^{j_{\ell}+1}$, it holds that 

\begin{equation*}
    \sum_{t=0}^{t_{\ell,2}-1}d(s_{t}(\widehat{\OPT}_{j_{\ell}}),s_{t+1}(\widehat{\OPT}_{j_{\ell}}))\leq\widehat{\OPT}_{j_{\ell}}<2^{j_{\ell}+1},
\end{equation*}

where $s_{0}(\widehat{\OPT}_{j_{\ell}})$ is defined to be the initial state $s_{0}$. Comparing two trajectories
\begin{equation*}
    s_{0},s_{1}(\widehat{\OPT}_{j_{\ell}}),\ldots,s_{t_{\ell,2}}(\widehat{\OPT}_{j_{\ell}}) \text{ and } s_{0},s_{t_{\ell,1}}(\widehat{\OPT}_{j_{\ell}}),\ldots,s_{t_{\ell,2}}(\widehat{\OPT}_{j_{\ell}}),
\end{equation*}
it is easy to see that the latter is a shortcut of the former, and hence the total length of the latter is no more than that of the former, which is no more than $2^{j_{\ell}+1}$ based on the definition of a phase.
\end{proof}
\begin{lemma}
    The external cost between the $\ell$'th and the $(\ell+1)$'th phase is no more than $2^{j_{\ell}+1}$ for every $1\leq \ell\leq k-1$.
\end{lemma}

\begin{proof}
    If $t_{\ell,2}$ is the last step of the $\ell$'th phase and the Guess-and-Double Method moves to the $(\ell+1)$'th phase at the next step, then the external cost between the two phases is $d(s_{0},s_{t_{\ell,2}}(\widehat{\OPT}_{j_{\ell}})$. Since the whole trajectory $s_{0},s_{1}(\widehat{\OPT}_{j_{\ell}}),\ldots,s_{t_{\ell,2}}(\widehat{\OPT}_{j_{\ell}})$ has length no more than $2^{j_{\ell}+1}$, the shortcut $d(s_{0},s_{t_{\ell,2}}(\widehat{\OPT}_{j_{\ell}})$ has length no more than $2^{j_{\ell}+1}$ based on triangle inequalities.
\end{proof}

\begin{theorem}
    The strict competitive ratio of Algorithm~\ref{alg2} on every two-set chasing instance is no more than $6$.
\end{theorem}

\begin{proof}
    The internal cost of each phase with the guessed value $j_{\ell}$ is at most $2^{j_{\ell}+1}$, and the total internal cost is no more than 

    \begin{equation*}
        \sum_{\ell=1}^{k}2^{j_{\ell}+1}\leq \sum_{\ell=1}^{k}2^{j_{k}+1-k+\ell}=2^{j_{k}+1}\sum_{t=1}^{k}\frac{1}{2^{t-1}}< 2^{j_{k}+2}.
    \end{equation*}

Similarly, the total external cost is no more than $2^{j_{k-1}+2}\leq 2^{j_{k}+1}$. The algorithmic execution terminates in phase $k$ with guessed value $j_{k}$, so the final optimal solution is at least $2^{j_{k}}$, and \cref{alg2} achieves strict competitive ratio of no more than $\frac{2^{j_{k}+2}+2^{j_{k}+1}}{2^{j_{k}}}=6$.
\end{proof}

\begin{remark}
The Guess-and Double Method for the two-set chasing problem (parametrized metrical service system with $m=2$) can be generalized to the case where every possible request sequence is an unknown prefix of some fixed sequence that the algorithm knows in advance. In this case, the only uncertainty is the length of the actual input sequence, and $\widehat{\text{OPT}}_{j}$ for every $j\geq 0$ can be computed prior to the execution of the algorithm. Although the research of metrical task systems has been focused on finite discrete metrics (and the majority of the classic results express the competitive ratio as a function of the number of points $n$), this Guess-and Double Method actually works for every general metric space including the Euclidean space etc., since it only requires that every $\widehat{\text{OPT}}_{j}$ can be computed prior to the execution of the algorithm, which is guaranteed by that every possible request sequence is some prefix of a given infinite sequence, and that the triangle inequalities of the metric space hold.
\end{remark}

%% file: Reference.bib
@inproceedings{DBLP:conf/approx/BubeckR20,
  author       = {S{\'{e}}bastien Bubeck and
                  Yuval Rabani},
  editor       = {Jaroslaw Byrka and
                  Raghu Meka},
  title        = {Parametrized Metrical Task Systems},
  booktitle    = {Approximation, Randomization, and Combinatorial Optimization. Algorithms
                  and Techniques, {APPROX/RANDOM} 2020, August 17-19, 2020, Virtual
                  Conference},
  series       = {LIPIcs},
  volume       = {176},
  pages        = {54:1--54:14},
  publisher    = {Schloss Dagstuhl - Leibniz-Zentrum f{\"{u}}r Informatik},
  year         = {2020},
  url          = {https://doi.org/10.4230/LIPIcs.APPROX/RANDOM.2020.54},
  doi          = {10.4230/LIPICS.APPROX/RANDOM.2020.54},
  bibsource    = {dblp computer science bibliography, https://dblp.org}
}

@inproceedings{DBLP:conf/stoc/BorodinLS87,
  author       = {Allan Borodin and
                  Nathan Linial and
                  Michael E. Saks},
  editor       = {Alfred V. Aho},
  title        = {An Optimal Online Algorithm for Metrical Task Systems},
  booktitle    = {Proceedings of the 19th Annual {ACM} Symposium on Theory of Computing,
                  1987, New York, New York, {USA}},
  pages        = {373--382},
  publisher    = {{ACM}},
  year         = {1987},
  url          = {https://doi.org/10.1145/28395.28435},
  doi          = {10.1145/28395.28435},
  bibsource    = {dblp computer science bibliography, https://dblp.org}
}

@article{DBLP:journals/siamcomp/BubeckCLL21,
  author       = {S{\'{e}}bastien Bubeck and
                  Michael B. Cohen and
                  James R. Lee and
                  Yin Tat Lee},
  title        = {Metrical Task Systems on Trees via Mirror Descent and Unfair Gluing},
  journal      = {{SIAM} J. Comput.},
  volume       = {50},
  number       = {3},
  pages        = {909--923},
  year         = {2021},
  url          = {https://doi.org/10.1137/19M1237879},
  doi          = {10.1137/19M1237879},
  bibsource    = {dblp computer science bibliography, https://dblp.org}
}

@article{DBLP:journals/tecs/IraniSG03,
  author       = {Sandy Irani and
                  Sandeep K. Shukla and
                  Rajesh K. Gupta},
  title        = {Online strategies for dynamic power management in systems with multiple
                  power-saving states},
  journal      = {{ACM} Trans. Embed. Comput. Syst.},
  volume       = {2},
  number       = {3},
  pages        = {325--346},
  year         = {2003},
  url          = {https://doi.org/10.1145/860176.860180},
  doi          = {10.1145/860176.860180},
  bibsource    = {dblp computer science bibliography, https://dblp.org}
}

@inproceedings{DBLP:conf/stoc/BubeckCR23,
  author       = {S{\'{e}}bastien Bubeck and
                  Christian Coester and
                  Yuval Rabani},
  editor       = {Barna Saha and
                  Rocco A. Servedio},
  title        = {The Randomized k-Server Conjecture Is False!},
  booktitle    = {Proceedings of the 55th Annual {ACM} Symposium on Theory of Computing,
                  {STOC} 2023, Orlando, FL, USA, June 20-23, 2023},
  pages        = {581--594},
  publisher    = {{ACM}},
  year         = {2023},
  url          = {https://doi.org/10.1145/3564246.3585132},
  doi          = {10.1145/3564246.3585132},
  bibsource    = {dblp computer science bibliography, https://dblp.org}
}

@article{kleitman1975dedekind,
  title={On Dedekind’s problem: the number of isotone Boolean functions. II},
  author={Kleitman, Daniel and Markowsky, George},
  journal={Transactions of the American Mathematical Society},
  volume={213},
  pages={373--390},
  year={1975}
}

@article{DBLP:journals/cacm/SleatorT85,
  author       = {Daniel Dominic Sleator and
                  Robert Endre Tarjan},
  title        = {Amortized Efficiency of List Update and Paging Rules},
  journal      = {Commun. {ACM}},
  volume       = {28},
  number       = {2},
  pages        = {202--208},
  year         = {1985},
  url          = {https://doi.org/10.1145/2786.2793},
  doi          = {10.1145/2786.2793},
  bibsource    = {dblp computer science bibliography, https://dblp.org}
}

@article{DBLP:journals/jal/ManasseMS90,
  author       = {Mark S. Manasse and
                  Lyle A. McGeoch and
                  Daniel Dominic Sleator},
  title        = {Competitive Algorithms for Server Problems},
  journal      = {J. Algorithms},
  volume       = {11},
  number       = {2},
  pages        = {208--230},
  year         = {1990},
  url          = {https://doi.org/10.1016/0196-6774(90)90003-W},
  doi          = {10.1016/0196-6774(90)90003-W},
  bibsource    = {dblp computer science bibliography, https://dblp.org}
}

@inproceedings{DBLP:conf/focs/Yao77,
  author       = {Andrew Chi{-}Chih Yao},
  title        = {Probabilistic Computations: Toward a Unified Measure of Complexity
                  (Extended Abstract)},
  booktitle    = {18th Annual Symposium on Foundations of Computer Science, Providence,
                  Rhode Island, USA, 31 October - 1 November 1977},
  pages        = {222--227},
  publisher    = {{IEEE} Computer Society},
  year         = {1977},
  url          = {https://doi.org/10.1109/SFCS.1977.24},
  doi          = {10.1109/SFCS.1977.24},
  bibsource    = {dblp computer science bibliography, https://dblp.org}
}

@article{sperner1928satz,
  title={Ein satz {\"u}ber untermengen einer endlichen menge},
  author={Sperner, Emanuel},
  journal={Mathematische Zeitschrift},
  volume={27},
  number={1},
  pages={544--548},
  year={1928},
  publisher={Springer}
}

@inproceedings{DBLP:conf/dimacs/ChrobakL91,
  author       = {Marek Chrobak and
                  Lawrence L. Larmore},
  editor       = {Lyle A. McGeoch and
                  Daniel Dominic Sleator},
  title        = {The Server Problem and On-Line Games},
  booktitle    = {On-Line Algorithms, Proceedings of a {DIMACS} Workshop, New Brunswick,
                  New Jersey, USA, February 11-13, 1991},
  series       = {{DIMACS} Series in Discrete Mathematics and Theoretical Computer Science},
  volume       = {7},
  pages        = {11--64},
  publisher    = {{DIMACS/AMS}},
  year         = {1991},
  url          = {https://doi.org/10.1090/dimacs/007/02},
  doi          = {10.1090/DIMACS/007/02},
  bibsource    = {dblp computer science bibliography, https://dblp.org}
}

@article{DBLP:journals/algorithmica/BurleyI97,
  author       = {William R. Burley and
                  Sandy Irani},
  title        = {On Algorithm Design for Metrical Task Systems},
  journal      = {Algorithmica},
  volume       = {18},
  number       = {4},
  pages        = {461--485},
  year         = {1997},
  url          = {https://doi.org/10.1007/PL00009166},
  doi          = {10.1007/PL00009166},
  bibsource    = {dblp computer science bibliography, https://dblp.org}
}

@article{DBLP:journals/siamcomp/AugustineIS08,
  author       = {John Augustine and
                  Sandy Irani and
                  Chaitanya Swamy},
  title        = {Optimal Power-Down Strategies},
  journal      = {{SIAM} J. Comput.},
  volume       = {37},
  number       = {5},
  pages        = {1499--1516},
  year         = {2008},
  url          = {https://doi.org/10.1137/05063787X},
  doi          = {10.1137/05063787X},
  bibsource    = {dblp computer science bibliography, https://dblp.org}
}

@inproceedings{DBLP:conf/focs/FiatFKRRV91,
  author       = {Amos Fiat and
                  Dean P. Foster and
                  Howard J. Karloff and
                  Yuval Rabani and
                  Yiftach Ravid and
                  Sundar Vishwanathan},
  title        = {Competitive Algorithms for Layered Graph Traversal},
  booktitle    = {32nd Annual Symposium on Foundations of Computer Science, San Juan,
                  Puerto Rico, October 1-4, 1991},
  pages        = {288--297},
  publisher    = {{IEEE} Computer Society},
  year         = {1991},
  url          = {https://doi.org/10.1109/SFCS.1991.185381},
  doi          = {10.1109/SFCS.1991.185381},
  bibsource    = {dblp computer science bibliography, https://dblp.org}
}

@inproceedings{DBLP:conf/soda/Ramesh93,
  author       = {H. Ramesh},
  editor       = {Vijaya Ramachandran},
  title        = {On Traversing Layered Graphs On-line},
  booktitle    = {Proceedings of the Fourth Annual {ACM/SIGACT-SIAM} Symposium on Discrete
                  Algorithms, 25-27 January 1993, Austin, Texas, {USA}},
  pages        = {412--421},
  publisher    = {{ACM/SIAM}},
  year         = {1993},
  url          = {http://dl.acm.org/citation.cfm?id=313559.313843},
  bibsource    = {dblp computer science bibliography, https://dblp.org}
}

@article{DBLP:journals/jal/Burley96,
  author       = {William R. Burley},
  title        = {Traversing Layered Graphs Using the Work Function Algorithm},
  journal      = {J. Algorithms},
  volume       = {20},
  number       = {3},
  pages        = {479--511},
  year         = {1996},
  url          = {https://doi.org/10.1006/jagm.1996.0024},
  doi          = {10.1006/JAGM.1996.0024},
  bibsource    = {dblp computer science bibliography, https://dblp.org}
}

@article{DBLP:journals/mor/BuchbinderN09,
  author       = {Niv Buchbinder and
                  Joseph Naor},
  title        = {Online Primal-Dual Algorithms for Covering and Packing},
  journal      = {Math. Oper. Res.},
  volume       = {34},
  number       = {2},
  pages        = {270--286},
  year         = {2009},
  url          = {https://doi.org/10.1287/moor.1080.0363},
  doi          = {10.1287/MOOR.1080.0363},
  bibsource    = {dblp computer science bibliography, https://dblp.org}
}
